\documentclass[a4paper,12pt]{article}
\usepackage{graphicx}
\usepackage{amsfonts,amsmath,amsthm,amssymb}
\usepackage{multirow}
\usepackage{booktabs}
\usepackage{tikz}
\usepackage{fullpage}
\usepackage{mathtools} 
\usepackage{soul}
\usepackage{float}
\usepackage{verbatim}
\usepackage{ulem}

\newcommand{\A}{\mathcal{A}}

\newcommand{\F}{\mathcal{F}}

\newcommand{\R}{\mathbb{R}}

\newtheorem{theorem}{Theorem}

\newtheorem{step}{Step}

\def\A{\boldsymbol{A}}

\def\R{\mathbb{R}}

\def\frac#1#2{{#1\over#2}}

\def\boxit#1{\vbox{\hrule\hbox{\vrule\kern.75truemm
\vbox{\kern.75truemm#1\kern1truemm}\kern1truemm\vrule}\hrule}}
\def\epsilon{\varepsilon}

\def\boxit#1{\vbox{\hrule\hbox{\vrule\kern.75truemm
\vbox{\kern.75truemm#1\kern1truemm}\kern1truemm\vrule}\hrule}}

\numberwithin{equation}{section}

\def\A{{\cal A}}

\newtheorem{lemma}{Lemma}

\newtheorem{example}{Example}
\newtheorem*{ex1cont*}{Example 1 - cont}

\begin{document}

\title{Gain-Loss Hedging and Cumulative Prospect Theory}
\author{Lorenzo Bastianello, Alain Chateauneuf, Bernard Cornet\footnote{Bastianello (corresponding author): Università Ca' Foscari Venezia, (email: lorenzo.bastianello@unive.it); Chateauneuf: IPAG Business School, Universit\'{e} Paris 1 Panth\'{e}on-Sorbonne and Paris School of Economics, (email: chateaun@univ-paris1.fr ); Cornet: Universit\'{e} Paris 1 Panth\'{e}on-Sorbonne and Kansas University (email: cornet@ku.edu). First version: February 2020
}}
\maketitle

\begin{abstract}

Two acts are comonotonic if they yield high payoffs in the same states of nature. The main purpose of this paper is to derive a new characterization of Cumulative Prospect Theory (CPT) through simple properties involving comonotonicity.  The main novelty is a concept dubbed gain-loss hedging: mixing  positive and negative acts creates hedging possibilities even when  acts  are comonotonic. This allows us to clarify in which sense CPT differs from Choquet expected utility. Our analysis is performed under the simpler case of (piece-wise) constant marginal utility which allows us to clearly separate the perception of uncertainty from the evaluation of outcomes.



\end{abstract}

\section{Introduction}\label{sec:intro}

When making everyday decisions, economic agents are often confronted with uncertainty. For instance, one can think of a decision maker (DM) who needs to choose how to allocate her wealth between two different portfolios of assets, or a firm that has to decide whether to invest in an innovative technology or in a traditional one. The most popular model used under risk and uncertainty is the expected utility model. This model, proposed first by Bernoulli at the beginning of the XVIII century,  has been axiomatized by de Finetti \cite{deFinetti}, Savage \cite{Savage} and Von Neumann and Morgenstern \cite{vNM}. However, empirical evidence has shown that expected utility does not provide a good description of DMs' actual choices. Early examples are the famous paradoxes of Allais \cite{Allais} and Ellsberg \cite{Ellsberg}. 

One of the most prominent and most successful alternative to  expected utility theory is cumulative prospect theory (CPT) of Tversky and Kahneman \cite{CPT}. The aim of this paper is twofold: $(i)$ we provide a new mathematical characterization of the CPT functional under the simplifying assumption of (piece-wise) constant marginal utility \textit{à la} Yaari~\cite{Yaari}; $(ii)$ we use the characterization of the previous point to obtain a novel preference axiomatization of CPT.

Consider acts as functions from a state space $S$ to the set of real numbers. Thus, given an act $f:S\rightarrow\R$, $f(s)$ can be interpreted as the amount of money or consumption good that a DM obtains if the state turns out to be $s$. A central role is played by comonotonic acts.  Loosely speaking, two acts are comonotonic if they are positively correlated. Mixing two comonotonic acts does not provide a possible hedge against uncertainty. This idea was exploited in the seminal papers of Schmeidler \cite{Schmeidler86}, \cite{Schmeidler89} to extend expected utility to Choquet expected utility. 

One advantage of CPT over the Choquet model is that  it allows to disentangle the behavior  of DMs in the domain of gains from the one in the domain of losses, i.e. when outcomes are respectively above or below a certain reference point (in our case the reference point is naturally taken equal to 0). This difference in behavior can be decomposed into two components. The first one is called loss-aversion and says that ``losses loom larger than gains'' (Tversky and Kahneman \cite{CPT}). Mathematically, it means that  losses are multiplied by a constant $\lambda>1$. The second one is usually called sign dependence and says that the attitude toward uncertainty (mathematically represented by a capacity) is different for gains and for losses. We take this behavior as a starting point for both the mathematical characterization of CPT and its axiomatization. The intuition behind our properties is that adding comonotonic acts can still provide some hedge if those acts are of opposite signs and have non-disjoint supports. We call this property gain-loss hedging.

We describe here the two main properties that we use  in Section \ref{sec:CPT_math} to characterize mathematically the  CPT functional. The first property is well-known and postulates comonotonic independence (separately) for gains and for losses.  Comonotonic acts do not provide a possible hedge against uncertainty and therefore adding them should not change the preferences of the DM. Take three acts $f,g,$ and $h$, all in the domain of gains or all in the domain of losses, such that $h$ is comonotonic with $f$ and $g$. Then  our condition require that if $f$ and $g$ are indifferent, then adding $h$ to both of them does not change a DM's preferences since in both situations $h$ does not increase nor reduce uncertainty.  
 
The second property, that we call gain-loss hedging, represents the main behavioral novelty of the paper. The key idea is that adding an act above the reference point to an act below the reference point may provide an hedge against uncertainty unless these acts have disjoint supports. To exemplify suppose that there are two states of the world $S=\{s_1,s_2\}$ and that a DM  with a linear utility function over outcomes is indifferent between the assets $f=(20,0)$  ($f$ is the act that pays 20 if $s_1$ is realized and 0 otherwise) and $g=(10,10)$. Consider now the act $h=(0,-10)$ which has disjoint support with $f$ but not with $g$. When the DM evaluates $f+h=(20,-10)$ and $g+h=(10,0)$, she may feel $f+h$ more uncertain than $g+h$ and therefore she may prefer $g+h$. Note that indifference between $f$ and $g$ and then a strict preference for $g+h$ is precluded by the expected  utility model (with the utility function being the identity). More interestingly, this preference pattern would be a paradox even for the more general Choquet expected utility model of Schmeidler \cite{Schmeidler89} (with the utility function being the identity). The Choquet   model excludes  any possible hedging through mixing of comonotonic acts. In this example however act $h$ is comonotonic with both acts $f$ and $g$ and therefore no possible hedging would be envisioned by the  Choquet model. Therefore $h$ is a possible hedge to uncertainty when added to $g$ because gains and losses balance out one another, and not because of comonotonicity. We elaborate more on this idea in Example \ref{ex:gain-loss-hedge}.

In Section \ref{sec:CPT_behave}, we give a preference axiomatization of the CPT model with piece-wise linear utility. We do not assume the Anscombe and Aumann \cite{AA} framework, and our axioms only appeals to simple properties related to comonotonicity. Moreover, we propose a new and simple axiom that can be used to elicit the coefficient of loss-aversion $\lambda$. In order to derive a CPT representation of preferences, we use the mathematical characterization of Section \ref{sec:CPT_math}. In a sense, our paper parallels, in the context of prospect theory, the work of Schmeidler \cite{Schmeidler86}, \cite{Schmeidler89} on the Choquet integral.

Empirical evidence not only supports sign-dependence, but it suggests further that agents are uncertainty averse for gains and uncertainty seeking for losses, see for instance Wakker~\cite{WakkerPT}, Section 12.7 for a review. Section \ref{sec:CPT_ambig_att} provides testable axioms that characterize those opposite behaviors. Finally we investigate when uncertainty aversion for gains is symmetric to uncertainty seeking for losses. Behaviorally, this happens if a DM who is indifferent between an act $f$ and a monetary outcome $\alpha$ is also indifferent between $-f$ and $-\alpha$. In this case we prove that weights for gains and losses are dual with respect to each other and that CPT reduces to a \v Sipo\v s integral, see \v Sipo\v s \cite{Sipos}. This result clarifies the relation of CPT with the \v Sipo\v s integral that was first noticed  by Starmer and Sudgen \cite{Starmer89} (see also Wakker \cite{WakkerPT} and Kothiyal \textit{et al.} \cite{KSW}).

 Of course, there are several axiomatizations of CPT available in the literature. The concept of comonotonicity and the fact that acts are rank-ordered are crucial, see Diecidue and Wakker \cite{10-2Wakker}. The very first axiomatization is provided in the seminal paper of Tversky and Kahneman \cite{CPT} and relies on comonotonic independence and a property called double matching. See also Trautmann and Wakker \cite{TW} for a recent characterization using these axioms in a (reduced) Anscombe and Aumann \cite{AA} framework. Wakker and Tversky \cite{WT} pair comonotonicity with trade-off consistency (see also Chateauneuf and Wakker \cite{ChatoWakker99} for the case of risk). The conomotonic sure thing principle approach (or a weakening of it called tail independence) is developed in Chew and Wakker \cite{ChewWakker96}, Zank \cite{Zank} and Wakker and Zank \cite{WZ}. The paper closest to our is the one of Schmidt and Zank \cite{SchmidtZank09}. The authors characterize CPT through an axiom called independence of common increments for comonotonic acts. Interestingly, they obtain a piecewise linear utility function (with a kink about the reference point), as in our axiomatization. We refer the reader to the introductory section of Schmidt and Zank \cite{SchmidtZank09} for a detailed discussion about the advantages of adopting piece-wise linear utility.

The rest of the paper is organized as follows. Section \ref{sec:framework} introduces the framework, the mathematical notations and the behavioral models that we will consider. Section \ref{sec:main} is divided in three subsections and it contains our main results. Section \ref{sec:CPT_math} presents the mathematical characterization of the CPT and \v Sipo\v s functionals, Section \ref{sec:CPT_behave} provides a behavioral characterization of CPT and Section \ref{sec:CPT_ambig_att} discusses  DM's attitude towards uncertainty. Section \ref{sec:conclusion} concludes. All proofs are gathered in the Appendix.

\section{Framework}\label{sec:framework}

Let $S$ be a set of states of the world endowed with a $\sigma$-algebra $\A$. Elements of $\A$ are called \textit{events}. We denote $\F$  the set of all bounded, real-valued, $\A$-measurable functions over $S$, i.e. $\F~=~\{f:S\rightarrow \R| f \text{ is bounded and }\A\text{-mesurable}\}$. A function $f\in \F$ is called \textit{act}. An act can be interpreted as an asset that pays  a monetary outcome in $\R$ that depends on the realization of the state of the world. We denote the \textit{positive part} of an act $f\in\F$ by $f^+=f\vee 0$ and the \textit{negative part} by $f^-=(-f)\vee 0$. Note that both positive and negative parts are greater than 0.\footnote{Note that several papers studying prospect theory use the symbol $f^-$ to denote $f\wedge 0$.}
The set $\F^+=\{f\in \F| f(s)\geq 0,\, \forall s\in S\}$ is the set of positive acts, the set of negative acts $\F^-$ is defined analogously. Two acts $f,g\in \F$ have the \textit{same sign} if either $f,g\in \F^+$ or $f,g\in \F^-$. We say that two acts are of \textit{opposite sign} if one of them is positive and the other is negative. The \textit{support} of an act of $f\in \F$ is the set $supp(f)=\{s\in S | f(s)\neq 0\}$.  Two acts $f,g\in \F$ are \textit{comonotonic} if for all $s,t\in S$, $(f(s)-f(t))(g(s)-g(t))\geq 0$.
Let $A\subseteq S$, $1_A$ is the \textit{indicator function} of the set $A$, i.e. $1_A(s):= \begin{cases}
1 & \text{ if } s\in A\\
0 & \text{ if } s\in A^c
\end{cases}$.
If $\alpha\in \R$, then $\alpha 1_A$ denotes the constant act which pays $\alpha$ in every state $s\in A$.

A \textit{(normalized) capacity} $v$ on the  measurable space $(S,\A)$ is a set function $v:\A \mapsto \mathbb{R}$ such that $v(\emptyset)=0,\,v(S)=1$ and for all $A,B \in \A,\, A\subseteq B \Rightarrow v(A)\leq v(B)$. If $v$ is a capacity, we define its \textit{conjugate} by $\hat{v}(A)=1-v(A^c)$ for all $A\in \A$.  A capacity $v: \A \mapsto \mathbb{R}$ is \textit{convex (concave)} if, for all $A,B\in \A$, $v(A\cup B)+v(A\cap B)\geq(\leq) v(A)+v(B)$. 
Given a capacity $v$ on $(S,\A)$, the \textit{Choquet integral} of $f\in \F$ with respect to $v$ is a functional $C:\F\rightarrow \R$ defined by
$$
C(f)=\int_{S}{f\,dv}:=\int_{-\infty}^0{(v(\{f\geq t\})-1)\,dt} + \int_0^{+\infty}{v(\{f\geq t\})\,dt},
$$
In the following we will remove the subscript $S$ from the integral sign whenever the domain of integration is clear. Given a capacity $v$ on $(S,\A)$, the \textit{\v Sipo\v s integral} (see Sipo\v s~\cite{Sipos}) of  $f\in \F$ with respect to $v$ is  a functional $\check{S}:\F\rightarrow \R$  defined as
$$
\check{S}(f)=\int f^+ dv-\int f^- dv
$$
where the two integrals are Choquet integrals.  The following Lemma 
gives an alternative formulation of the \v Sipo\v s integral when the conjugate capacity is used when evaluating the negative part of a function. Moreover it clarifies the relation between the Choquet integral and the \v Sipo\v s integral.

\begin{lemma}\label{lemma:sipos_ceu}
Let $v$ be a capacity and $\hat{v}$  its conjugate. Then the following holds:
\begin{itemize}
\item $\check{S}(f)= \int f^+ dv+\int- f^- d\hat{v}$;
\item $C(f)=\int f^+ dv+\int- f^-dv=\int f^+dv -\int f^- d\hat{v}$.
\end{itemize}
\end{lemma}

The main object of this paper is the \textit{(piecewise linear) Cumulative Prospect Theory (CPT)} functional $CPT:\F\rightarrow \R$. It is a generalization of both Choquet and \v Sipo\v s integrals. Consider two capacities $v^+$, $v^-$ and a real number $\lambda>0$, then the \textit{(piecewise linear) CPT} functional $CPT:\F\rightarrow \R$  is defined by
$$
CPT(f)=\int f^+ dv^+-\int \lambda f^- dv^-.
$$

A preference relation $\succsim$ over $\F$ is a complete and transitive binary relation with non-empty strict part. As usual, $f\succsim g$ means ``$f$ is preferred to $g$''. We denote $\succ$ and $\sim$ the strict and weak part of $\succsim$. A functional $I:\F\rightarrow\R$ \textit{represents} $\succsim$ if for all $f,g\in \F$, $f\succsim g$ if and only if $I(f)\geq I(g)$.

\section{Main results}\label{sec:main}

This section contains our two main results. The first one, Theorem \ref{th:CPT}, characterizes mathematically the CPT functional.  The second result, Theorem \ref{th:axiom_CPT}, studies which behavioral axioms a preference relation should satisfy in order to be represented by a CPT functional. 

\subsection{The CPT functional}\label{sec:CPT_math}

We start with a seminal theorem of Schmeidler \cite{Schmeidler86} who provided a characterization of the Choquet functional. Before presenting the result we recall that a functional $I:\F\rightarrow \R$  is \textit{monotonic} if $f\geq g\Rightarrow I(f)\geq I(g)$, where $f\geq g$ means $f(s)\geq g(s)$ for all $s\in S$. Moreover $I$ satisfies \textit{comonotonic additivity} if, whenever $f$ and $g$ are comonotonic, then  $I(f+g)=I(f)+I(g)$.

\begin{theorem}\label{th:Schmeidelr86}{\sc (Schmeidler \cite{Schmeidler86})} Let $I:\F\rightarrow \R$ be a given functional with $I(1_S)=1$. Then the following are equivalent.
\begin{itemize}
\item[(i)] $(a)$ $I$ is monotonic; $(b)$ $I$  satisfies comonotonic additivity.
\item[(ii)] $I$ is a Choquet integral.
\end{itemize}

\end{theorem}

The CPT functional generalizes the Choquet functional by relaxing comonotonic additivity. More specifically, comonotonic additivity will be retained only for comonotonic acts of the same sign and for (comonotonic) acts of opposite sign with disjoint supports. The following is our first main result.

\begin{theorem}\label{th:CPT}  Let $I:\F\rightarrow \R$ be a given functional satisfying $I(1_S)=1$. Then the following are equivalent.
\begin{itemize}
\item[(i)] $(a)$ $I$ is monotonic; $(b)$ $I$  satisfies comonotonic additivity on $\F^+$ and $\F^-$ and for acts $f,g$ of opposite sign such that $supp(f)\cap supp(g)=\emptyset$. 
\item[(ii)]  $I$ is a CPT functional.
\end{itemize}
\end{theorem}

Consider item $(i)$ of both Theorem \ref{th:Schmeidelr86} and Theorem \ref{th:CPT}. Note that part (b) of Theorem \ref{th:Schmeidelr86} implies (b) of Theorem \ref{th:CPT}, as acts with opposite sign and disjoint supports are comonotonic. This relaxation not only characterize a functional that is more general than the Choquet integral, but it also gives some important insights from a behavioral point of view.

Recall that comonotonic additivity is a weakening of full-fledged  additivity, a property that would force the functional $I$ to be linear, and hence an expectation. The behavioral intuition behind comonotonic additivity is that adding two comonotonic acts does not permit possible hedging against choices of nature. 

Relaxing comonotonic additivity allows us to uncover more sophisticated attitudes towards uncertainty and more subtle forms of hedging. The fist remarkable property of the CPT functional is that it differentiates agents' behavior in the domain of gains (i.e. $\F^+$) from the one in the domain of losses (i.e. $\F^-$). The outcome for which behavior changes, namely the monetary outcome 0, is called the \textit{reference point}.\footnote{In this paper, the reference point is exogenously given and it is normalized to 0 for convenience (we could have chosen any other reference point $r\in \R$). Schmidt and Zank \cite{SchmidtZank12} provide axioms to make the reference point endogenous.} Comonotonic additivity is preserved whenever acts under considerations are both above or both below the reference point. Comonotonic additivity over  $\F^+$ and  $\F^-$ weakens a condition already well known in the literature called cosigned independence. Two acts $f,g\in \F$ are \textit{sign-comonotonic} or simply \textit{cosigned} if they are comonotonic and there exists no $s\in S$ such that $f(s)>0$ and $g(s)<0$, see Wakker and Tversky \cite{WT} and Trautmann and Wakker~\cite{TW}.

One of the main contribution of the present paper lies in the second comonotonic additivity requirement that characterizes the $CPT$ functional, namely $f,g$ of opposite sign such that $supp(f)\cap supp(g)=\emptyset$ implies $CPT(f+g)=CPT(f)+CPT(g)$. This means that comonotonic additivity can fail if we have $f,g$ of opposite sign and $supp(f)\cap supp(g)\neq\emptyset$ (we underline again that such acts are comonotonic). The behavioral intuition behind this requirement is that adding the positive and negative parts of two acts can provide a hedge against possible choices of nature even when acts under consideration are comonotonic. We call this property \textit{gain-loss hedging}.\footnote{We thank Peter Wakker for suggesting this terminology.} This hedging possibility it is not considered for instance in the Choquet model, where the only way to hedge is to add two non-comonotonic acts. The following example provides more details for the particular case in which CPT reduces to a \v Sipo\v s integral, i.e.  $\lambda=1$ and $v^+=v^-$.

\begin{example}\label{ex:gain-loss-hedge}
Let $S=\{s_1,s_2,s_3\}$ and consider a CPT functional with $\lambda=1$ and  $v=v^+=v^-$ (i.e. a \v Sipo\v s integral). Let $v$ be defined as

\begin{center}
\begin{tabular}{c|c|c|c|c|c|c|c|c}
$A$ & $S$ & $\emptyset$ & $s_1$ & $s_2$  & $s_3$ & $s_1\cup s_2$ & $s_2\cup s_3$ & $s_1\cup s_3$   \\ \hline
$v$ & 1 & 0 & $\frac{2}{3}$ & $\frac{1}{3}$ & 0 & $\frac{2}{3}$ & $\frac{2}{3}$ &  1  \\ 
\end{tabular}
\end{center}
Consider now the following acts on $S$.
\begin{center}
\begin{tabular}{c|c|c|c}
	 &  $s_1$ & $s_2$  & $s_3$    \\ \hline
 $f$ &  3 & 4  & 4   \\ 
 $g$ &  0 & 11  & 0   \\ 
 $h$ &  -3 & 0  & -1   \\ 
 $-h$ &  3 & 0  & 1   \\ 
 $f+h$ &  0 & 4  & 3   \\ 
 $g+h$ &  -3 & 11  & -1   \\ 
\end{tabular}
\end{center}
Acts $f,g,h$ are comonotonic, but $supp(g)\cap supp(h)=\emptyset$ while $supp(f)\cap supp(h)\neq\emptyset$. Let $\succsim_{\check{S}}$ be the preference relation induced by the $\check{S}$ functional, i.e. $f\succsim_{\check{S}} g \Leftrightarrow \check{S}(f)\geq \check{S}(g)$ and  $\succsim_C$ the one induced by the $C$ functional (both functionals $\check{S}$ and $C$ are defined in Section \ref{sec:framework}). We have
\begin{align*}
\check{S}(f)=C(f)= & 3+(4-3)\frac{2}{3}=\frac{11}{3} \\
\check{S}(g)=C(g)= & 0+(11-0)\frac{1}{3}=\frac{11}{3}
\end{align*}
and therefore $f\sim_{\check{S}}g$ and $f\sim_C g$. Moreover since $h$ is comonotonic with $f$ and $g$, by comonotonic additivity  $f+h\sim_C g+h$ (one can actually verify that $C(f+h)=C(f)+C(h)=\frac{7}{3}=C(g)+C(h)=C(g+h)$). However we can notice that the act $f+h$ looks much  ``smoother'' than $g+h$ and moreover $f+h\geq0$ since gains balance losses. This intuition is captured by the preference relation induced by the \v Sipo\v s integral as
\begin{align*}
\check{S}(f+h)= & 0+(3-0)\frac{2}{3}+(4-3)\frac{1}{3}=\frac{7}{3} \\
\check{S}(g+h)= & C(g)- C(-h)= \frac{11}{3}- (0+(1-0)1+(3-1)\frac{2}{3})=\frac{4}{3}
\end{align*}
and therefore $f+h\succ_{\check{S}} g+h$.

\end{example}

Example \ref{ex:gain-loss-hedge} shows that gain-loss hedging is an interesting behavioral feature of CPT and of \v Sipo\v s integrals. Adding positive and negative acts with supports that are not disjoint, can provide an hedge even when the acts involved are comonotonic. This happens because gains compensate losses.
 In the next section we provide a new behavioral foundation of CPT taking this observation as a starting point. 

Example \ref{ex:gain-loss-hedge} shows that preferences represented by  \v Sipo\v s integrals are rich enough to entail gain-loss hedging behaviors. It is therefore interesting to mathematically characterize Sipo\v s integrals. Theorem \ref{th:Sipos} shows that a symmetric condition pins down a CPT functional as a \v Sipo\v s integral.

\begin{theorem}\label{th:Sipos}  A CPT functional is a \v Sipo\v s integral if and only if $CPT(-f)=-CPT(f)$ for all $f\in \F$.
\end{theorem}

Theorem \ref{th:Sipos} says that CPT reduces to  a \v Sipo\v s integral if and only if the condition $CPT(-f)=-CPT(f)$ for all $f\in \F$ is satisfied. This is an interesting result as such condition is a strong one. As an example,  if $C$ is a Choquet functional then $C(-f)=-C(f)$ for all $f\in \F$ if and only if the capacity $v$ equals its conjugate $\hat{v}$, and therefore it is additive on events $\{A,A^c\}$.

\subsection{A behavioral characterization of CPT}\label{sec:CPT_behave}

In this section we provide a preference axiomatization of CPT.  We recall that a preference relation $\succsim$ over $\F$ is a complete and transitive binary relation with non-empty strict part.

The first axiom is a continuity axiom. 

\medskip
\noindent \textsc{A.1 Continuity.} The sets $\{\alpha\in\R| \alpha1_S\succsim f\}$ and $\{\alpha\in\R| f\succsim  \alpha1_S\}$ are closed for all $f\in \F$. 
\medskip

Note that the axiom requires only to compare acts with constants. This dispenses us to formulate topological assumptions on the set of acts $\F$.

The second axiom is a monotonicity property.

\medskip
\noindent \textsc{A.2 Monotonicity.}  Let $f,g\in \F$  be such that $f \geq g$. Then $f\succsim g$.
\medskip

Consider now the well known comonotonic independence axiom (Chateauneuf \cite{Chato94}, Schmeidler \cite{Schmeidler89}). It says that if two acts $f$ and $g$ are indifferent to each other, then adding a comonotonic act $h$  to both of them does not change the DM's preferences. The idea behind this condition is that adding comonotonic acts does not provide any possible hedge against uncertainty.

\medskip
\noindent \textsc{A.C Comonotonic Independence.}  Let $f,g,h\in \F$ such that $h$ is comonotonic with $f$ and with $g$. Then $f\sim g$ implies $f+h\sim g+h$ 
\medskip

Preferences satisfying A.1, A.2 and A.C are represented by a Choquet integral. We present this result in the next proposition

\begin{theorem}{\sc (Chateauneuf \cite{Chato94}, Schmeidler \cite{Schmeidler89})}
Let $\succsim$ be a preference relation over $\F$. Then the following are equivalent.
\begin{itemize}
\item[(i)] $\succsim$ satisfies A.1, A.2 and A.C.
\item[(ii)]  There exists a (unique) capacity $v$ such that $\succsim$ is represented by a Choquet functional.
\end{itemize}
\end{theorem}

However, as Example \ref{ex:gain-loss-hedge} shows, Comonotonic Independence may be too strong as it doesn't take into account (gain-loss) hedging possibilities that arise adding positive and negative acts with non-disjoint supports. The following two axioms, axiom A.3 and A.4, are both implied by Comonotonic Independence. They are at the heart of our behavioral characterization. They generalizes A.C in two directions. First, axiom A.3 allows for different attitudes towards uncertainty in the domain of gain and in the domain of losses. Second, axiom A.4 takes into account possible gain-loss hedging opportunities that arises in situations like the one of Example \ref{ex:gain-loss-hedge}.

\medskip

\noindent \textsc{A.3 Comonotonic Independence for Gain and Losses.}  Let $f,g,h\in \F^{+(-)}$ be such that $h$ is comonotonic with $f$ and $g$. Then $f\sim g$ implies $f+h\sim g+h$.

\medskip

\noindent \textsc{A.4 $\lambda$-Disjoint Independence.} There exists $\lambda>0$ such that for all $f\in \F^{+}$ and $g\in \F^{-}$ such that $supp(f)\cap supp(g)=\emptyset$ and such that $f\sim \alpha1_S$ and $g\sim \beta 1_S$
\begin{enumerate}
\item if $\alpha+\lambda \beta\geq 0$ then $f+g\sim (\alpha+\lambda \beta)1_S$;
\item if $\alpha+\lambda \beta<0$ then $f+g\sim \left(\frac{\alpha+ \lambda\beta }{\lambda} \right) 1_S$.
\end{enumerate}

Axiom A.4 represents the main behavioral novelty. To better understand it, note that it is implied by the following (stronger) axiom.

\medskip

\noindent \textsc{A.4$^*$ Disjoint Independence.} For all $f\in \F^{+}$ and $g\in \F^{-}$ such that $supp(f)\cap supp(g)=\emptyset$ and such that $f\sim \alpha1_S$ and $g\sim \beta 1_S$, one has $f+g\sim (\alpha+\beta)1_S$.

\medskip

It is easy to see that A.4$^*$ follows from A.4 imposing $\lambda=1$. A.4$^*$  requires that the act $f+g$ is evaluated as the sum of its constant equivalent. In the general case, we can have $\lambda\neq 1$ and in this case A.4 says that the constant equivalent of $f+g$ depends on the sign of $\alpha+\lambda \beta$. The interpretation for the case of loss-aversion, $\lambda>1$, is the following. The DM outweighs losses by a factor of $\lambda$ and considers $\lambda\beta$ instead of $\beta$ \textit{tout-court}. If $\alpha+\lambda \beta>0$ then the DM feels ``overall in the domain of gains'' and the certainty equivalent of $f+g$ is positive and such that $\alpha>0$ is balanced by $\lambda \beta<\beta<0$, i.e. the  certainty equivalent $\beta$ of losses is outweighed by a factor of $\lambda$. If $\alpha+\lambda \beta<0$ then the DM feels ``overall in the domain of losses'' and in this case the certainty equivalent of $f+g$ is negative and  equal to $\beta<0$ plus $\frac{\alpha}{\lambda}>0$, i.e. the certainty equivalent $\alpha$ of the positive part decreased by a factor of $\lambda$ (since $\frac{\alpha}{\lambda}<\alpha$).

Importantly, $\lambda$ can be determined in the lab:  take $f\in\F^+$ and $g\in\F^-$ such that $supp(f)\cap supp(g)=\emptyset$, ask the certainty equivalents $\alpha$, $\beta$ and $\gamma$ of $f$, $g$ and $f+g$ respectively. If $\gamma=\alpha+\beta$, there is no loss-aversion or seeking. If $\gamma\neq\alpha+\beta$, then if $\gamma>0$ we have $\lambda=\frac{\gamma-\alpha}{\beta}$ and if $\gamma<0$ we have $\lambda=\frac{\alpha}{\gamma-\beta}$. There is lively debate on whether loss-aversion is a real phenomena or not, with results on both sides. See for instance Gal and Rucker \cite{Gal} and Gächter, Johnson and Herrmann \cite{Gat}. We hope therefore that A.4 could be helpful to elicit loss-aversion in a setting in which individuals' preferences are represented by the CPT functional with piece-wise constant marginal utility.

When A.C is replaced by A.3 and A.4, we obtain a characterization of the CPT functional. The following is our second main result.

\begin{theorem}\label{th:axiom_CPT}
Let $\succsim$ be a preference relation over $\F$. Then the following are equivalent.
\begin{itemize}
\item[(i)] $\succsim$ satisfies A.1, A.2, A.3 and A.4.
\item[(ii)]  There exist two (unique) capacities $v^+$, $v^-$ and a real number $\lambda>0$ such that $\succsim$ is represented by a CPT functional.
\end{itemize}
\end{theorem}

Note that if we replace A.4 with A.4$^*$ in Theorem \ref{th:axiom_CPT}, we obtain a CPT functional with $\lambda=1$, i.e. loss-neutrality.

\subsection{Attitude towards uncertainty}\label{sec:CPT_ambig_att}

As we already said above, a remarkable property of CPT is that (unlike the Choquet functional) it allows to disentangle DMs' attitude towards uncertainty in the domain of gains from the one in the domain of losses. This is made possible since an act is evaluated through the sum of two Choquet integrals with respect to a capacity $v^+$ for gains, and a different one $v^-$ for losses. 

Experimental evidence shows that  DMs are uncertainty averse for gains and uncertainty seeking for losses. Loosely speaking, uncertainty aversion (seeking) means that agents prefer situations in which objective probabilities of events are (not) available. In our framework, objective probabilities are not there at all. Therefore an act is not uncertain only if it is a constant act. Intuitively, in our purely subjective setting, an uncertainty averse (seeking) DM would prefer acts that are ``as close (far) as possible'' to constant acts. We capture this idea with the two following axioms. 

\medskip
\noindent A.3' Let $f,g,h\in \F^+$  such that $h$ is comonotonic with $g$. Then $f\sim g \Rightarrow f+h\succsim g+h$. 
\medskip

\medskip
\noindent A.3'' Let $f,g,h\in \F^-$  such that $h$ is comonotonic with $f$. Then $f\sim g \Rightarrow f+h\succsim g+h$. 
\medskip

Axiom  A.3' captures the intuition that DMs are uncertainty averse in the domain of gains $\F^+$. Consider three acts $f,g,h\in \F^+$ such that $f\sim g$ and $h$ is comonotonic with $f$ and $g$. Then adding (the potentially non-comonotone act) $h$ to $f$ increases the appreciation of $f$ since $h$ may by an hedge against $f$, while at the same time it decreases the appreciation of $g$ since uncertainty may be higher. To exemplify, let $A\in \A$ and consider  $f=10 \cdot 1_A+5\cdot 1_{A^c}$, $g=5 \cdot 1_A+10 \cdot 1_{A^c}$ and $h=0 \cdot 1_A+5\cdot 1_{A^c}$. Then $f+h=10 \cdot 1_S$ is a constant act while $g+h=5 \cdot 1_A+15 \cdot 1_{A^c}$ is even more uncertain than $g$. A DM who dislikes uncertainty would clearly prefer $f+h$ to   $g+h$. 
Axiom A.3'' can be interpreted similarly, but in this case the DM is willing to increase the perceived uncertainty. Notice that similar conditions were proposed by Chateuneuf \cite{Chato94}, see also Wakker \cite{Wakker90}.

The following theorem shows that if a DM is uncertainty averse for gains and uncertainty seeking for losses then the capacities appearing in the CPT functional are both convex.

\begin{theorem}\label{th:axiom_CPT_convex}
Let $\succsim$ be a preference relation over $\F$. Then the following are equivalent.
\begin{itemize}
\item[(i)] $\succsim$ satisfies A.1, A.2, A.3', A.3'', and A.4.
\item[(ii)]  There exist two (unique) convex capacities $v^+$, $v^-$ and $\lambda>0$, such that $\succsim$ is represented by a CPT functional.
\end{itemize}
\end{theorem}

Note that the CPT functional can be rewritten (using Lemma \ref{lemma:basic} in the Appendix) as
\begin{equation}\label{eq:CPT_formula2}
CPT(f)=\int f^+ dv^++\int -\lambda f^- d\hat{v}^-.
\end{equation}
If one is using this formulation then Theorem \ref{th:axiom_CPT_convex} implies that the conjugate capacity $\hat{v}^-$  is concave. 

We conclude this section providing a testable axiom that characterizes symmetric attitudes around the reference point with respect to uncertainty. An interesting question is in fact to understand when one has $v^-=v^+$ in the CPT functional.\footnote{Or, equivalently if one is using formulation (\ref{eq:CPT_formula2}), when one gets $v^-=\hat{v}^+$.} Note that if $\lambda=1$ Theorem \ref{th:Sipos} applies and one gets a \v Sipo\v s integral. Consider the following axiom.

\medskip
\noindent \textsc{A.5 Gain-Loss Symmetry.} Let $f\in \F$ and $\alpha\in \R$. Then $f\sim \alpha1_S$ if and only if $-f\sim -\alpha1_S$.
\medskip

Axiom A.5 says that if a DM is indifferent between an (uncertain) act $f$ and a sure amount $\alpha$, then she should stay indifferent between $-f$ and $-\alpha$. The intuition is that the DM sees $f$ and $-f$ as symmetric with respect to the reference point 0, and therefore evaluates them through the symmetric sure amounts $\alpha$ and $-\alpha$. The following theorem offers a behavioral characterization of the \v Sipo\v s integral.

\begin{theorem}\label{th:axiom_Sipos}
Let $\succsim$ be a preference relation over $\F$. Then the following are equivalent.
\begin{itemize}
\item[(i)] $\succsim$ satisfies A.1, A.2, A.3, A.4 and A.5.
\item[(ii)]  There exists a (unique)  capacity $v$  such that $\succsim$ is represented by  the  \v Sipo\v s integral.
\end{itemize}
\end{theorem}

\section{Conclusion}\label{sec:conclusion}

We provided an axiomatic analysis of CPT with piece-wise linear utility. This allowed us to focus on (sign-dependent) attitudes towards uncertainty. 
First, we mathematically characterized the CPT functional by weakening the comonotonic additivity property of the Choquet integral. We also gave conditions to reduce CPT to a \v Sipos integral. Then we gave an axiomatic characterization of CPT.  The main novelty is given by a gain-loss hedging property: gains and losses balance each other out and provide an hedge against uncertainty. Moreover, we introduced an axiom that offers a way to easily elicit the coefficient of loss-aversion, in case of piece-wise linear utility. Finally, we characterized uncertainty aversion for losses and uncertainty loving for gains. Moreover we showed that these attitudes are symmetric with respect to the reference point if and only if CPT is  a \v Sipos integral.

\newpage


\appendix
\section{Appendix}
\small

We begin with an elementary Lemma. The proof if given for sake of completeness.

\begin{lemma}\label{lemma:basic}
Let $\hat{v}(A)=1-v(A^c)$ and $f\in \F^+$ or $f\in \F-$. Then $-\int f dv = \int - fd \hat{v}$.
\end{lemma}
\begin{proof}
Let $f\in \F^+$, then one has
\begin{align*}
-\int f dv &= -\int_{0}^{\infty}v(s\in S|f(s)\geq t) dt \\
	&= - \int_{0}^{\infty}[1-\hat{v}(s\in S|f(s)> t) ]dt \\
	&= - \int_{0}^{-\infty}[1-\hat{v}(s\in S|f(s)> u)](-du) \\
	&= \int^{0}_{-\infty}[\hat{v}(s\in S|f(s)> u)-1]du \\
	&= \int - fd \hat{v}
\end{align*}
The case $f\in \F^-$ can be treated similarly.
\end{proof}

\medskip
\medskip

\begin{proof}[\textbf{Proof of Lemma \ref{lemma:sipos_ceu}}]
To prove the first point we just need to apply Lemma \ref{lemma:basic}. In fact noticing that $f^-\in \F^+$ we have $\check{S}(f)=\int f^+ dv-\int f^- dv=\int f^+ dv+\int - f^-d \hat{v}.$\\
The second point, note that $f=f^++(-f^-)$ and  that $f^+$ and $-f^-$ are comonotonic.
Then by the comonotonic additivity of the Choquet integral proved in Theorem \ref{th:Schmeidelr86}, we have 
$$
\int f dv=\int f^++(-f^-) dv=\int f^+dv +\int -f^- dv.
$$
Note that by Lemma \ref{lemma:basic} one can also write $\int f dv=\int f^+dv -\int f^- d\hat{v}$.
\end{proof}

\medskip
\medskip

\begin{proof}[\textbf{Proof of Theorem \ref{th:CPT}}]
$(i) \Rightarrow (ii)$. We start with an auxiliary Lemma.
\begin{lemma}\label{lem:useful}
For all $\alpha>0$, for all $f\in \F^+\cup \F^-$, $I(\alpha f)=\alpha I(f)$. Moreover for $\alpha>0$ and $f\in \F^+$, $I(f+\alpha1_S)=I(f)+\alpha$.
\end{lemma}
\begin{proof}[Proof of Lemma \ref{lem:useful}]
The proof is standard.
\end{proof}
Let $v^+(A)=I(1_A)$, then doing the same proof as Schmeidler \cite{Schmeidler86} one can show that for all $f\in \F^+$, $I(f)=\int f dv^+$.
Now let $\lambda:=-I(-1_S)$. By comonotonic additivity of $I$, $I(0)=0$. By monotonicity of $I$, $I(-1_S)\leq I(0)=0$. Then $\lambda\geq 0$. Define for all $A\in \A$, 
$$
v(A)=-\frac{I(-1_A)}{\lambda}.
$$
We have $v(\emptyset)=0$ and $v(S)=1$. Take $A\subseteq B$ so that $-1_A\geq-1_B$. Since $I$ is monotonic, $I(-1_A)\geq I(-1_B)$ and therefore $v(A)\leq v(B)$. This show that $v$ is a capacity.  Define $v^-$ as the conjugate capacity of $v$, meaning that for all $A\in \A$,
$$
v^-(A)=1-v(A^c).
$$
We will show that for all $f\in \F^-$, $f$ simple, $I(f)=\int \lambda f dv^-$. Let $f\in \F^-$ be defined as
$$
f=x_11_{A_1}+\dots +x_n1_{A_n}
$$
where $\{A_1,\dots,A_n\}$ is a partition of $S$ and $x_1\leq \dots \leq x_n \leq 0$. Note that we can rewrite $f$ as
$$
f=(0-x_n)(-1_S)+(x_n-x_{n-1})(-1_{A_{n-1}\cup\dots\cup A_1})+\dots+(x_3-x_2)(-1_{A_2\cup A_1})+(x_2-x_1)(-1_{A_1}).
$$
Define 
$$
h_i=(x_{i+1}-x_{i})(-1_{A_{i}\cup\dots\cup A_1})
$$
with the convention that $x_{n+1}=0$. We have that 
$$
f=\sum_{i=0}^n h_i
$$

We show now that $h_i$ is comonotone with $\sum_{k=i+1}^n h_k$. Consider $s,t\in S$ be such that $s\in A_{i}\cup\dots\cup A_1$ and $t\in (A_{i}\cup\dots\cup A_1)^c$, suppose $t\in A_l$ for $l>i$. Then $h_i(s)-h_i(t)=x_i-x_{i+1}\leq 0$ and $\sum_{k=i+1}^n h_k(s)-\sum_{k=i+1}^n h_k(t)=x_i-x_l\leq 0$, hence $(h_i(s)-h_i(t))\left(\sum_{k=i+1}^n h_k(s)-\sum_{k=i+1}^n h_k(t)\right)\geq 0$. If $s,t\in A_{i}\cup\dots\cup A_1$ or $s,t\in (A_{i}\cup\dots\cup A_1)^c$ the previous product is 0. This shows that the functions $h_i$ and $\sum_{k=i+1}^n h_k$ are conomotone. 

Since $h_i$ and $\sum_{k=i+1}^n h_k$ are negative, by comonotonic additivity on $\F^-$ we have 
$$
I(f)=I(h_1+\sum_{i=2}^n h_i)=I(h_1)+I(h_2+\sum_{i=3}^n h_i)=\dots= \sum_{i=1}^n I(h_i). 
$$

Note that by Lemma \ref{lem:useful} and by definition of $v^-$ we have
$$
I(h_i)=(x_{i+1}-x_{i})\lambda(v^-(A_{i+1}\cup\dots\cup A_n)-1).
$$
Therefore 
\begin{align*}
I(f)&= \sum_{i=1}^{n-1} I(h_i) +I(h_n) \\
	&=  \lambda \left[\sum_{i=1}^{n-1}(x_{i+1}-x_{i})v^-(A_{i+1}\cup\dots\cup A_n)+\sum_{i=1}^{n-1}(x_{i+1}-x_{i})\right]+(x_{n+1}-x_n)(-\lambda)\\
	&= \lambda \left[\sum_{i=1}^{n-1}(x_{i+1}-x_{i})v^-(A_{i+1}\cup\dots\cup A_n)-x_n+x_1\right]+\lambda x_n \\
	&=\lambda \left[x_1+\sum_{i=1}^{n-1}(x_{i+1}-x_{i})v^-(A_{i+1}\cup\dots\cup A_n)\right]\\
	&=\lambda\int f dv^- \\
	&=\int  \lambda f dv^-
\end{align*}

Notice that every bounded function can be approximated  by a sequence of step functions as in Schmeidler \cite{Schmeidler86}. This shows that for all $f\in \F^-$, $I(f)=\int \lambda f dv^+$. Let now $f\in \F$ and notice that $f=f^+ +(-f^-)$ and moreover $supp(f^+)\cap supp(f^-)=\emptyset$. Hence 
$$
I(f)=I(f^++(-f^-))=I(f^+)+I(-f^-)=\int f^+dv^++\int \lambda(-f^-)d\bar{v}^-
$$
Let $\hat{v}^-$ be the conjugate capacity of $\bar{v}^-$, i.e. $\hat{v}^-(A)=1-\bar{v}^-(A^c)$ for all $A\in \A$. Then by Lemma \ref{lemma:basic} one has
$$
\int -\lambda f^-d\bar{v}^-=-\int \lambda f^-d\hat{v}^-
$$
Defining $v^-=\hat{v}^-$ concludes the ``$\Rightarrow$'' part of the proof.

$(ii)\Rightarrow (i)$. We prove (a). Suppose $f\geq g$. Then $f^+\geq g^+$ and $g^-\geq f^-$. It is well known that the Choquet integral is monotonic. Hence $I(f)=\int f^+dv^+ - \int\lambda f^-dv^-\geq g^+dv^+ - \int\lambda g^-dv^-=I(g)$ . \\
We prove (b). Let $f,g$ comonotonic and such that $f,g\geq 0$ (the case $f,g\leq 0$ is similar). Then $(f+g)^+=f+g=f^++g^+$ and $(f+g)^-=0=f^-=g^-$. Therefore 
\begin{multline*}
I(f+g)=\int(f+g)^+dv^+=\int f^+dv^+ + \int g^+dv^+= \\
\int f^+dv^+ - \int\lambda f^-dv^- + \int g^+dv^+ - \int\lambda g^-dv^-=I(f)+I(g).
\end{multline*}
We prove part (b) of $(ii)$. Let $f$, $g$ be of opposite sign (for instance $f\geq0$ and $g\leq 0$) and such that $supp(f)\cap supp(g)=\emptyset$. Notice that $(f+g)^+=f=f^+$,  $(f+g)^-=-g=g^-$, and $f^-=0$, $g^+=0$. Therefore 
\begin{multline*}
I(f+g)=\int(f+g)^+dv^+-\int\lambda(f+g)^-dv^-=\int f^+dv^+ - \int\lambda g^-dv^-=  \\
\int f^+dv^+ - \int\lambda f^-dv^- + \int g^+dv^+ - \int \lambda g^-dv^-=I(f)+I(g)
\end{multline*}
which complete the proof of part (b).
\end{proof}

\medskip
\medskip

\begin{proof}[\textbf{Proof of Theorem \ref{th:Sipos}}]
$(i) \Rightarrow (ii)$ Let $CPT$ be a \v Sipo\v s integral. Then $\lambda=1$ and $v^+=v^-$ and hence
$$
CPT(-f)=\int(-f)^+ dv-\int (-f)^- dv=\int f^- dv-\int f^+ dv=-CPT(f)
$$

$(ii) \Rightarrow (i)$  Note that $\lambda=1$ since $-\lambda=CPT(-1_S)=-CPT(1_S)=-1$. Let $A\in \A$ and consider $f=1_A$.  Then 
$$
CPT(-f)=0-\int 1_A dv^-=-v^-(A) \text{ and } -CPT(f)=-\int 1_A dv^+=-v^+(A)
$$
Therefore 
$$
CPT(-f)=-CPT(f)\Leftrightarrow v^-(A)=v^+(A).
$$
Since this must be true for all $A\in \A$, $v^-=v^+$ and the CPT functional is a \v Sipo\v s integral.
\end{proof}

\medskip
\medskip

\begin{proof}[\textbf{Proof of Theorem \ref{th:axiom_CPT}}]
$(ii)\Leftarrow (i)$ We only prove A.4. Take $\lambda>0$ of the CPT functional. Fix $f$ and $g$ s.t. $f\in \F^{+}$ and $g\in \F^{-}$ and s.t. $supp(f)\cap supp(g)=\emptyset$. Suppose $f\sim \alpha1_S$ and $g\sim \beta 1_S$. Note that $\alpha\geq 0$ and $\beta\leq 0$. Therefore
\begin{align*}
CPT(f)=CPT(\alpha1_S) &\Leftrightarrow CPT(f)=\alpha \\
CPT(g)=CPT(\beta1_S) &\Leftrightarrow CPT(g)=-\lambda\int \beta^-dv^-=-\lambda(-\beta)=\lambda\beta.
\end{align*}
Moreover since $f$ and $g$ have opposite signs and have disjoints supports we have
$$
CPT(f+g)=CPT(f)+CPT(g)=\alpha+\lambda\beta.
$$ 
Now, if $\alpha+\lambda\beta>0$, $CPT((\alpha+\lambda\beta)1_S)=\alpha+\lambda\beta$, and since CPT represents $\succsim$, $f+g\sim  (\alpha+\lambda \beta)1_S$. If  $\alpha+\lambda\beta<0$, $CPT\left(\frac{\alpha+\lambda \beta}{\lambda}1_S\right)=-\lambda\int \left(\frac{\alpha+\lambda \beta}{\lambda}\right)^-dv^-=-\lambda\frac{-\alpha-\lambda \beta}{\lambda}=\alpha+\lambda\beta$. Therefore $f+g\sim \frac{\alpha+\lambda \beta}{\lambda}1_S$.


$(i)\Rightarrow(ii)$ First, note that for all $f\in \F^+$, $f=f^+$ and for all $f\in \F^-$, $f=-f^-$. (One can prove that) By A.1 and A.2 for all  $f\in \F^+$ there exists a unique $\alpha_{f^+}\geq 0$ s.t. 
$$
f^+\sim \alpha_{f^+}1_S.
$$
Let $\lambda>0$ be the one of Axiom A.4. Then again by A.1 and A.2 for all  $f\in \F^-$ there exists a unique $\alpha_{-f^-}\leq 0$ s.t. 
$$
-f^-\sim \left(\frac{\alpha_{-f^-}}{\lambda}\right)1_S.
$$
Define $I:\F\rightarrow\R$ as 
$$
I(f)=I(f^+)+I(-f^-)
$$
where $I(f^+)=\alpha_{f^+}$ and $I(-f^-)=\alpha_{-f^-}$. Note that $f^+\sim I(f^+)1_S$ and $-f^-\sim \left(\frac{I(-f^-)}{\lambda}\right)1_S$. Moreover $I(1_S)=1$ by Monotonicity.

We will prove that $I$ satisfies the conditions of Theorem \ref{th:CPT} and it is therefore a CPT functional.

\begin{step}\label{step:CE}
Fix $f\in\F$, then  $I(f)\geq0$ implies $f\sim I(f)1_S$ and $I(f)<0$ implies $f\sim \frac{I(f)}{\lambda} 1_S$. 
\end{step}
\begin{proof}
Let $f\in\F$. 
\begin{itemize}
\item Case 1: $I(f)\geq0$. Note that $f=f^++(-f^-)$ and by definition $f^+\sim I(f^+)1_S$ and $-f^-\sim \frac{I(-f^-)}{\lambda}1_S$. Moreover $I(f^+)+\lambda\frac{I(-f^-)}{\lambda}=I(f)\geq0$, hence by A.4 and by the definition of $I(f)$
$$
f=f^++(-f^-)\sim \left(I(f^+)+\lambda\frac{I(-f^-)}{\lambda}\right)1_S=I(f)1_S.
$$
\item Case 2: $I(f)<0$. Then reasoning as before and applying A.4 we get
$$
f=f^++(-f^-)\sim \left(\frac{I(f^+)+\lambda\frac{I(-f^-)}{\lambda}}{\lambda} \right)1_S=\left(\frac{I(f^+)+I(-f^-)}{\lambda}\right)1_S=\frac{I(f)}{\lambda} 1_S.
$$
\end{itemize}
\end{proof}

\begin{step}\label{step:mon}
$I$ is monotone.
\end{step}
\begin{proof}
Let $f,g\in\F$ be such that $f\geq g$. Then $f^+\geq g^+$ and $-f^-\geq -g^-$. Then by Monotonicity $f^+\succsim g^+$ and  $-f^-\succsim -g^-$. Then  by Step \ref{step:CE} $I(f^+)1_S\sim f^+\succsim g^+\sim I(g^+)1_S$ and  $\frac{I(-f^-)}{\lambda}1_S\sim -f^-\succsim -g^-\sim \frac{I(-g^-)}{\lambda}1_S$.
 Monotonicity implies $I(f^+)\geq I(g^+)$ and $I(-f^-)\geq I(-g^-)$. Summing up we obtain $I(f)\geq I(g)$. 
\end{proof}

\begin{step}\label{step:como_add}
$I$ satisfies comonotonic additivity over $\F^+$ and $\F^-$.
\end{step}
\begin{proof}
We prove comonotonic additivity over $\F^-$, the proof for $\F^+$ can be done in a similar way. \\
Take $f,g\in\F^-$ s.t. $f$ and $g$ are comonotone. By Step $\ref{step:CE}$, $f\sim \frac{I(f)}{\lambda}1_S$ and $g\sim \frac{I(g)}{\lambda}1_S$. Since constant acts are comonotone with all other acts and $\frac{I(f)}{\lambda},\frac{I(g)}{\lambda}\leq 0$, by A.3 one gets $f+g\sim \frac{I(f)}{\lambda}1_S+g$ and $g+\frac{I(f)}{\lambda}1_S\sim \frac{I(g)}{\lambda}1_S+\frac{I(f)}{\lambda}1_S$. Since $f+g\in\F^-$, by Step \ref{step:CE} $f+g\sim \frac{I(f+g)}{\lambda}1_S$. Therefore $\frac{I(f+g)}{\lambda}1_S\sim \left(\frac{I(f)}{\lambda}+ \frac{I(g)}{\lambda}\right)1_S$, and Monotonicity implies $I(f+g)=I(f)+I(g)$.
\end{proof}

\begin{step}\label{step:disj_add}
For all $f\in \F^{+(-)}$ and $g\in \F^{-(+)}$ s.t. $supp(f)\cap supp(g)=\emptyset$, $I(f+g)=I(f)+I(g)$. 
\end{step}
\begin{proof}
Fix $f\in \F^{+}$ and $g\in \F^{-}$ s.t. $supp(f)\cap supp(g)=\emptyset$. Define $h=f+g$ and note that $h^+=f$ and $-h^-=g$. Therefore by definition of $I$,  $I(f+g)=I(h)=I(h^+)+I(-h^-)=I(f)+I(g)$.
%
\end{proof}

\begin{step}\label{step:represent}
$I$ represents $\succsim$ over $\F$ (i.e. $f\succsim g \Leftrightarrow I(f)\geq I(g)$).
\end{step}
\begin{proof}
Fix $f,g\in\F$. Then we have to consider 4 cases. 
\begin{itemize}
\item Case 1: $I(f),I(g)\geq0$. Using Step \ref{step:CE} and Monotonicity $I(f)1_S\sim f\succsim g\sim I(g)1_S \Leftrightarrow I(f)\geq I(g)$. 
\item Case 2: $I(f),I(g)\leq0$. Using Step \ref{step:CE} and Monotonicity $\frac{ I(f)}{\lambda}1_S\sim f\succsim g\sim \frac{ I(g)}{\lambda}1_S\Leftrightarrow I(f)\geq I(g)$, since $\lambda>0$.
\item Case 3: $I(f)\geq 0>I(g)$. Using Step \ref{step:CE} and Monotonicity $I(f)1_S\sim f\succsim g\sim \frac{ I(g)}{\lambda}1_S\Leftrightarrow I(f)\geq I(g)$. Note that in this case we cannot have $g\succsim f$.
\item Case 4: $I(g)\geq 0>I(f)$. This is the same as Case 3.
\end{itemize}
\end{proof}

Since $I(1_S)=1$, Steps \ref{step:mon}, \ref{step:como_add} and \ref{step:disj_add} prove that $I$  satisfies condition $(i)$ of Theorem \ref{th:CPT} and therefore $I$ is a CPT functional. Moreover Step \ref{step:represent} shows that $I$ represents $\succsim$. Therefore the proof is complete.
\end{proof}

\medskip
\medskip

\begin{proof}[\textbf{Proof of Theorem \ref{th:axiom_CPT_convex}}]
 $(i)\Rightarrow (ii)$ Note that A.3' and A.3'' imply A.3. Hence Theorem~\ref{th:axiom_CPT} applies and $I$ is represented by a CPT functional. It is left to show that $v^+$ and $v^-$ are convex. We only show convexity of $v^-$. Fix $A,B\in\A$ and note that $CPT(-1_A)=-\lambda v^-(A)=CPT(-v^-(A)1_S)$ and a similar statement holds for $B\in \A$. Therefore $-1_A\sim -v^-(A)1_S$ and $-1_B\sim -v^-(B)1_S$. Since $-1_B$ is comonotonic with  $-v^-(A)1_S$, by A.3'' $-v^-(A)1_S-1_B\succsim -1_A-1_B$.  Moreover since $-v^-(A)1_S$ is comonotonic with both $-1_B$ and $-v^-(B)1_S$ by A.3'' we get $-1_B-v^-(A)1_S\sim -v^-(B)1_S-v^-(A)1_S$. Therefore
$$
-v^-(B)1_S-v^-(A)1_S\sim -1_B-v^-(A)1_S\succsim -1_A-1_B
$$
Note that $-1_A-1_B=-1_{A\cup B}-1_{A\cap B}$ and since $1_{A\cup B}$ and $1_{A\cap B}$ are comonotonic, 
\begin{align*}
CPT(-1_{A\cup B}-1_{A\cap B})&=-\int \lambda (-1_{A\cup B}-1_{A\cap B})^-dv^-\\
&= -\lambda \left( \int1_{A\cup B} dv^-+  \int1_{A\cap B} dv^-\right) \\
&= -\lambda [v^-(A\cup B)+v^-(A\cap B)].
\end{align*}

Therefore $-\lambda [v^-(A)+v^-(B)]= CPT(-v^-(A)1_S-v^-(B)1_S)\geq CPT(-1_{A\cup B}-1_{A\cap B})=-\lambda [v^-(A\cup B)+v^-(A\cap B)]$ which implies $v^-(A)+v^-(B)\leq v^-(A\cup B)+v^-(A\cap B)$, i.e. $v^-$ is convex.

$(ii)\Rightarrow (i)$ Left to the reader.
\end{proof}

\medskip
\medskip

\begin{proof}[\textbf{Proof of Theorem \ref{th:axiom_Sipos}}]
$(i)\Rightarrow (ii)$ Since $\succsim$ satisfies A.1, A.2, A.3 and A.4, it can be represented by a CPT functional $I$ by Theorem \ref{th:axiom_CPT}. Hence for all $f\in \F$, $f\sim I(f)1_S$ and $-f\sim I(-f)1_S$. Notice that  by A.5 one has also $-f\sim -I(f)1_S$ and hence A.2 implies $I(-f)=-I(f)$. By Theorem \ref{th:Sipos}, $I$ is a \v Sipo\v s integral. 

$(ii)\Rightarrow (i)$ Left to the reader.
\end{proof}


\end{document}